\documentclass[copyright]{eptcs}
\providecommand{\event}{DCFS 2010} 
\usepackage{breakurl}             

\usepackage{amssymb}
\usepackage{amsthm}

\DeclareSymbolFont{letters}{OML}{cmm}{m}{it}
\DeclareSymbolFont{symbols}{OMS}{cmsy}{m}{n}

\newtheorem{proposition}{Proposition}
\newtheorem{definition}{Definition}

\title{Ciliate Gene Unscrambling with Fewer Templates}
\author{Lila Kari\thanks{This research was supported by the Natural Sciences and Engineering
Research Council of Canada Grant R2824A01, and Canada Research Chair Award
to L.K.}
\qquad\qquad Afroza Rahman
\institute{Department of Computer Science,
University of Western Ontario,\\
London, Ontario N6A 5B7, Canada}
\email{\qquad lila@csd.uwo.ca \quad\qquad arahman@cs.queensu.ca} 
}

\def\titlerunning{Gene Unscrambling with Fewer Templates}
\def\authorrunning{L.~Kari, A.~Rahman}

\begin{document}
\maketitle

\begin{abstract}
One of the theoretical models proposed for the mechanism of gene unscrambling in some species of ciliates is the \textit{template-guided recombination (TGR) system} by Prescott, Ehrenfeucht and Rozenberg which has been generalized by Daley and McQuillan from a formal language theory perspective. In this paper, we propose a refinement of this model that generates regular languages using the iterated TGR system with a finite initial language and a finite set of templates, using fewer templates and a smaller alphabet compared to that of the Daley-McQuillan model. To achieve Turing completeness using \textit{only} finite components, i.e., a  finite initial language and a finite set of templates, we also propose an extension of the \textit{contextual template-guided recombination system (CTGR system)} by Daley and McQuillan, by adding an extra control called \textit{permitting contexts} on the usage of templates.
\end{abstract}
\section{Introduction}

This paper proposes improvements in the descriptional complexity of two theoretical models of gene unscrambling in ciliates: template-guided recombination (TGR) systems, and contextual template-guided recombination (CTGR) systems. Ciliates are a group of unicellular eukaryotic protozoans, some of which have the distinctive characteristic of nuclear dualism, i.e., they have two types of nuclei: a functionally inert \textit{micronucleus} and an active \textit{macronucleus}. Genes in the active macronucleus provide RNA transcripts for the maintenance of the structure and function of the cell. Genes within the micronucleus are usually inactive and assist only in the conjugation process. The process of ``decrypting'' the micronuclear genes after conjugation, to obtain the functional macronuclear genes, is 
called {\em gene unscrambling} or {\em gene assembly}.

 The genes within micronuclear chromosomes are composed of protein-coding DNA segments (also known 
 as macronuclear destined sequences (\textit{MDSs})) interspersed by numerous, short, non-protein-coding DNA segments
 (also called internally eliminated sequences (\textit{IESs})). Furthermore, in some species of ciliates such as 
\textit{Oxytricha} or \textit{Stylonychia}, the micronuclear gene has been found to have a highly complex structure in which MDSs are stored in a permuted order. During the course of macronuclear development, 
these IESs are eliminated from the micronucleus by means of homologous recombination, and the permuted MDSs
 are sorted, resulting in a functionally complete macronucleus with MDSs present in the correct order.
In the micronuclear sequence, each MDS is flanked by guiding short sequences, 3 to 20 nucleotides long, 
 which act as pointers in a linked list. For instance, the
 $n$th MDS is flanked on the left by the same short sequence which flanks the $(n-1)$th MDS on the right.
 In the process of gene unscrambling, homologous recombination takes place between two DNA molecules that contain the identical
 guiding short sequences at the correct MDS-IES junctions.
 
Various theoretical models have been proposed in order to model the genetic unscrambling processes in ciliate organisms:
 the \textit{reversible guided recombination model}, \cite{Lil99, Lil63}, 
 based on binary inter- and intra-molecular DNA recombination operations; the
 \textit{ld, hi, dlad model}, \cite{Ter031, Ter03, Dav01}, based on three unary intra-molecular DNA recombination operations;
  the \textit{template-guided recombination (TGR) model}, \cite{Pre03}, where a
 DNA molecule from the old macronucleus conducts inter-molecular DNA recombination process serving as a template; 
 the \textit{RNA-guided DNA assembly model}, \cite{Ang07}, experimentally confirmed in \cite{Nowacki2007},
where either double-stranded RNA or single-stranded RNA act as templates.

This paper proposes two improvements of the descriptional complexity (size of template language, size of alphabet)
 of the template-guided recombination model as studied by Daley and McQuillan \cite{Dal05}.
 In \cite{Dal05}, it has been showed that the Daley-McQuillan 
TGR system can generate all regular languages using iterated template-guided recombination with finite initial
 and template languages. For this model, the considered gene rearrangement processes take place in
 a stochastic style \textit{in vivo} environment. In such a biological setting, it is significant
 to consider the size of the template language \cite{Domaratzki2009, Domaratzki2008}, because sufficient copies of each template must be available
 throughout the recombination process. This is essential to confirm the accessibility of a template
 in the right place at the proper time according to the demand. Hence, the number of the unique templates should
 be as low as possible. The first aim of this paper is a reduction of the size of the template language by
introducing a new approach to generate regular languages applying iterated TGR system with
 a finite initial language and a small finite set of templates. 

The second aim of this paper is the reduction of the size of the template language (from regular to finite),
in the extension of the template-guided recombination model called the 
 \textit{contextual template-guided recombination system (CTGR system)}, \cite{Mar06}. Recall that a CTGR is a TGR enhanced with 
  ``deletion contexts'', the introduction of which made it possible to enhance the TGR computational power to
 that of Turing machines. Our reason for wanting to achieve a reduction of the template set size
 is the obvious one, namely that handling an infinite regular set of templates in a biological setting is impossible.
 To achieve our goal, we employ an additional control over the templates in the
 form of ``permitting contexts''. We namely introduce the \textit{contextual template-guided recombination system
 (CTGR system) using permitting contexts} as an extension of the CTGR system,
 and prove that an iterated version of this system has the computational power of a Turing machine, but only
 uses a finite initial language and a finite set of templates.

The paper is organized as follows. Section~\ref{sec:TGR} introduces our new approach for generating the family
 of regular languages by using iterated TGR systems with $n^2$ templates, compared to $n^3$ templates in 
\cite{Dal05}. This reduction in descriptional complexity is achieved at the expense of using a filtering set
 to discard unintended results. Section \ref{sec:Extension_CTGR}
describes our proposed \textit{CTGR system using permitting contexts} that, unlike CTGR systems, are able to
 characterize the recursively enumerable languages by using only a finite base language and a 
{\em finite set of templates}. This reduction in the size of the template language
 is achieved by introducing an additional control mechanism, the permitting context, to CTGR.

We end this introduction by some formal definitions and notations.
An \textit{alphabet} is a finite and nonempty set of symbols. A \textit{word} or a \textit{string} is a finite sequence of symbols. Let $\Sigma$ be an alphabet. By $\Sigma^{\ast}$ we denote the set of all words over $\Sigma$ that includes the empty one denoted by $\lambda$. The set of nonempty words over $\Sigma$, i.e., $\Sigma^{\ast}$ $\setminus$ $\{\lambda\}$, is denoted by $\Sigma^{+}$. The length of a word $x \in \Sigma^{\ast}$ is denoted by $|x|$. For $k \in \mathbb{N}$, let $\Sigma^{\geq k} = \{w \mid w \in \Sigma^{\ast}, |w| \geq k\}$.
	
For two alphabets $\Sigma$, $\Delta$, a \textit{morphism} is a function $h : \Sigma^{\ast} \rightarrow \Delta^{\ast}$ satisfying $h(xy) = h(x)h(y)$ for all $x, y \in \Sigma^{\ast}$. A morphism $h : \Sigma^{\ast} \rightarrow \Delta^{\ast}$ is called a \textit{coding} if $h(a) \in \Delta$ for all $a \in \Sigma$ and a \textit{weak coding} if $h(a) \in \Delta \cup \{\lambda\}$.   
 We denote by \textit{RE}, \textit{CS}, \textit{CF}, \textit{LIN}, and \textit{REG} the families of languages generated by arbitrary, context-sensitive, context-free, linear, and regular grammars, respectively. By \textit{FIN} we denote the family of finite languages. For additional formal language theory definitions and notations the reader is referred to \cite{Salomaa1973}.

\section{TGR systems with fewer templates}
\label{sec:TGR}

This section proposes a refinement of the \textit{template-guided recombination (TGR)} model as studied in 
the model by Daley and McQuillan \cite{Dal05}, that is able to generate the family of regular languages
 by using a reduced number of templates and a smaller alphabet.

\begin{definition}
\label{TGR1}(\cite{Dal05})
A template-guided recombination system (or TGR system) is a four tuple\\ $\varrho = (T, \Sigma, n_{1}, n_{2})$ where $\Sigma$ is a finite alphabet, $T \subseteq \Sigma^{*}$ is the template language, $n_{1}$ is the minimum MDS length and $n_{2}$ is the minimum pointer length.
\end{definition}

For a TGR system $\varrho = (T,\Sigma,n_{1},n_{2})$ and a language $L\subseteq \Sigma^{*}$, $\varrho(L) = \{w\in\Sigma^{*}\mid(x,y)\vdash_{t}w$ for some $x,y \in L, t \in T \}$ where $(x,y)\vdash_{t}w$ iff $x = u\alpha\beta d, y = e\beta\gamma v, t = \alpha\beta\gamma, w = u\alpha\beta\gamma v, u, v, d, e \in \Sigma^{*},\alpha,\gamma \in \Sigma^{\geq n_1}, \beta \in \Sigma^{\geq n_2}$. 
 $L$ is sometimes called the \textit{base}, or \textit{initial} language.
 
Note that, if $x$ is a segment of the micronuclear DNA sequence that
contains the $n$th MDS $\alpha$,
and $y$ is a segment of the micronuclear DNA sequence that 
contains the $(n+1)$st MDS $\gamma$,
then the recombination between $x$ and $y$ guided by the template $t$ 
will result in bringing the MDSs $n$ and $(n+1)$ in the correct order in the intermediate 
DNA sequence $w$, regardless of their original position in the micronuclear sequence. A sequence
 of such template-guided recombinations is thought to accomplish the
 gene unscrambling, and the transformation of the micronuclear DNA sequence
in the macronuclear DNA sequence in ciliates.

For a TGR system $\varrho = (T,\Sigma,n_{1},n_{2})$ and a language $L \subseteq \Sigma^{*}$, 
 $\varrho^{*}(L)$ is defined as follows:
 $$\varrho^{0}(L) = L ,\quad \varrho^{n+1}(L)= \varrho^{n}(L) \cup \varrho(\varrho^{n}(L)),\;  n \geq 0,
\quad \varrho^{*}(L) = \bigcup^{\infty}_{n=0} \varrho^{n}(L).$$
If $\mathcal{L}_{1}, \mathcal{L}_{2}$ are two language families, then $\pitchfork^{*}(\mathcal{L}_{1},\mathcal{L}_{2},n_1,n_2) = \{\varrho^{*}(L)\mid L \in \mathcal{L}_{1}, \varrho = (T,\Sigma,n_{1},n_{2}), T \in \mathcal{L}_{2}\}$ and $\pitchfork^{*}(\mathcal{L}_{1},\mathcal{L}_{2}) = \{\pitchfork^{*}(\mathcal{L}_{1},\mathcal{L}_{2},n_1,n_2)\mid n_1,n_2 \in \mathbb{N}\}.$ 

In \cite[Prop.~15]{Dal05}, Daley and McQuillan prove that 
all regular languages can be generated using 
iterated template-guided recombination systems with finite initial and template languages, i.e., every regular language 
is a coding of a language in the family $\pitchfork^{*}(\mbox{FIN}, \mbox{FIN})$. The limitation of the Daley-McQuillan 
model \cite{Dal05} is that the size of the template language and the alphabet was not meant to be optimized. 
Since the size of the template language will have a great impact on this type of model during \textit{in vivo} computation, 
this is an important factor. Our aim is to reduce this number of templates.
We namely introduce a new approach to generate regular languages using iterated template-guided recombination, using a
 finite initial language, a finite set of templates, and a weak coding. 
We provide a simpler construction than that of \cite[Prop.~15]{Dal05}, with fewer templates and a smaller alphabet.    

\begin{proposition}
\label{TGRREGWE}
Each regular language $L \subseteq \Sigma^{*}$ can be written in the form $L = h(\varrho^{*}(L_{0}) \cap R)$, where $R$ is a regular language, $h$ is a weak coding homomorphism, $\varrho = (T, \Sigma^{'}, 1, 1 )$ is a TGR system, $T$ is a finite set of templates and $L_{0} \subseteq \Sigma^{'*}$ is a finite language.
\end{proposition}

\begin{proof}
Let $L \in REG$ be generated by a regular grammar $G = (N,\Sigma,S,P)$ with the rules in $P$ of the form
 $X \rightarrow aY$, $X \rightarrow a \mid \lambda$, for $X$, $Y \in N$, $a \in \Sigma$. We construct a TGR system
 $\varrho = (T, \Sigma^{'}, n_1, n_2 )$, where $n_1 = n_2 = 1$ and the alphabet
 $\Sigma^{'} = N \cup \Sigma \cup \{\#\}$. Here, $\#$ is a new symbol which assists to complete the recombination process acting
 as an end marker. Then we construct a finite \textit{base language} $L_{0} \subseteq \Sigma^{'*}$ and a \textit{template language} $T \subseteq \Sigma^{'*}$ as follows.
\begin{tabbing}
We \=define t\=he finite base language by:\\ 
			\> $L_1 = \{ S a \# \mid  \exists \:  S \rightarrow a \in P, a \in \Sigma\}$, $L_2 = \{ S a X \mid  \exists \:  S \rightarrow aX \in P, a \in \Sigma\}$,\\
			\> $L_3 = \{ X b Y \mid  \exists \: X \rightarrow bY \in P, X, Y \in N, b \in \Sigma\}$,\\
			\> $L_4 = \{ X a X \mid  \exists \: X \rightarrow aX \in P, X \in N, a \in \Sigma\}$, $L_5 = \{ X a \# \mid  \exists \: X \rightarrow a \in P\}$,\\
			\> $L_6 = \{ X \# \# \mid  \exists \: X \rightarrow \lambda \in P\}$, $L_0 = L_1 \cup L_2 \cup L_3 \cup L_4 \cup L_5 \cup L_6 $.\\
The\=~finite template language is \=defined by: \\ 
		\> $T_1 = \{a X b \mid a, b, X, Y, Z \in \Sigma^{'}, \exists \: Y \rightarrow aX \in P$, either
                       $\exists \: X \rightarrow bZ \in P$ or $\exists \: X \rightarrow b \in P \}$,\\
		\> $T_2 = \{a X a \mid a, X \in \Sigma^{'}, \exists \:  X \rightarrow aX \in P\}$,\\
		\> $T_3 = \{a X \# \mid a, X, Y, \# \in \Sigma^{'}, \exists \: Y \rightarrow aX \in P, \exists \: X \rightarrow \lambda \in P\}$, $T = T_1 \cup T_2 \cup T_3$.
\end{tabbing}

Note that for example, $L_4 \subseteq L_3$, $L_2 \subseteq L_3$, and $T_2 \subseteq T_1$. However, we made these separations for the purpose of the clarity of the proof.

 In order to eliminate all non terminals and the new symbol, we consider the weak coding $h$ defined by
$h(X) = \lambda$, for any $X \in N,$ $h(a) = a$, for any $a \in \Sigma,$  $h(\#) = \lambda $.
Moreover, we consider the language  $R = \{S\}(\Sigma N)^*\{\# ,\#\#\}$, whose purpose is to ensure
 that only strings of the correct form will be accepted, by removing other unintended strings.

We claim that $L = h(\varrho^{*}(L_{0}) \cap R)$.

For the ``$\subseteq$'' inclusion, in order to obtain a valid derivation in $G$ and to
 continue recombinations, we consider the string $S a X$ from group $L_2$ as the 
first string in the recombinations. At this stage, through recombination, the application of the
 rules of the form $X \rightarrow bY \in P$ can be achieved as follows. During the recombination,
 a string $XbY$ from group $L_3$ as the second string can be recombined with the string $S a X$, 
and an appropriate template $aXb$ from group $T_1$, can be used to produce the string $SaXbY$ which
 is of the form $\{S\}(\Sigma N)^*$, with $X, Y \in N$ and $a, b \in \Sigma$. By using only the 
templates from group $T_1$, the simulation of the rules of the form $X \rightarrow bY \in P$ 
is possible, because no other template from group $T_2$, $T_3$ can be used. The simulation
 is as follows:
$$(SaX, XbY) \vdash_{aXb} SaXbY.$$

The above mentioned simulation process can be repeated an arbitrary number of times according to the templates in group $T_1$. Likewise, rules of the form $X \rightarrow aX \in P$ can be simulated that produce the string $SaXaX$ using a template from $T_2$ and considering the second string from group $L_4$ as follows: 
$$(SaX, XaX) \vdash_{aXa} SaXaX.$$

The application of the rules of the form $X \rightarrow aX \in P$ can also be simulated repeatedly. In general, for representing an intermediate recombination, if $u, v \in \varrho^{*}(L_{0})$ illustrates derivations of $G$ and $u = u^{'}aY$, $v = Ybv^{'}$, where $u^{'} \in \{S\}(\Sigma N)^*$, $v^{'} \in (N\Sigma)^*N $, $a \in \Sigma$, $b \in \Sigma$, then the resulting recombined string $(u, v) \vdash_{aYb} u^{'}aYbv^{'}$ that is a string of the form $\{S\}(\Sigma N)^* $ can be generated which corresponds to an intermediate computation of the form $S \Rightarrow^* \delta N $ in $G$ where $\delta \in \Sigma^*$.

If a template finds more than one matching point in the first string, then the template can attach to any of those points and a matching second string from $L_0$ as guided by the template can be recombined with the first string. For example, such a recombination can happen to a string of the form\\ \hspace*{1.0 in}$Sa_1X_1a_2X_2 \ldots a_{i}X_{i}a_{i+1}X_{i+1} \ldots a_{k-1}X_{k-1}a_{k}X_{k}.$

Along this string if $a_{i}X_{i} = a_{k}X_{k}$ for some $1 \leq i \leq k$, then the recombination guided by a template $a_{i}X_{i}b$ can take place either at the matching position $a_{i}X_{i}$ or at the matching position $a_{k}X_{k}$ between the above first string and the second string of the form $X_ibY$. This recombination will produce the resulting string either of the form $Sa_1X_1a_2X_2 \ldots a_{i}X_{i}bY$ or of the form $Sa_1X_1a_2X_2 \ldots$ $a_{i}X_{i}a_{i+1}X_{i+1} \ldots$ $a_{k-1}X_{k-1}a_{k}X_{k}bY$, respectively. Note however that any recombination that does not happen at the rightmost end of the sentential form has only the effect of ``resetting" the derivation a few steps backward. Thus, without loss of generality, we will hereafter assume that any derivation that results in a terminal word has an equivalent rightmost derivation. We will only discuss these rightmost derivations.

Note also that recombinations can proceed in parallel, for example, a recombination can take place between a string of the form $Sa_1X_1a_2X_2 \ldots a_{i}X_{i}$ and a string of the form $X_{i}a_{i+1}X_{i+1}\ldots a_{k-1}X_{k-1}a_{k}X_{k}$ or alternatively $X_{i}a_{i+1}X_{i+1}\ldots$ $a_{k-1}X_{k-1}a_{k}\#$ using an appropriate template of the form $a_{i}X_{i}a_{i+1}$ that will lead to a resulting string of the form
 $Sa_1X_1a_2X_2 \ldots a_{i}X_{i}a_{i+1}X_{i+1} \ldots a_{k-1}X_{k-1}a_{k}X_{k}$ or\\ 
 $Sa_1X_1a_2X_2 \ldots a_{i}X_{i}a_{i+1}X_{i+1} \ldots a_{k-1}X_{k-1}a_{k}\#$, respectively. Any such derivation, however, can be replaced by a derivation that starts from a word containing $S$ and proceeds unidirectionally towards a terminal word. 

Let us now examine the simulation of the terminating rules. Here, it is assumed that a string of the form $Sa_1X_1a_2X_2 \ldots a_{n-1}X_{n-1} \in \{S\}(\Sigma N)^*$ is to be considered as the first string that was produced at the previous step. Now the application of the rules of the form $X_{n-1} \rightarrow a_n \in P$ can be achieved using the second string of the form $X_{n-1} a_n \#$ from group $L_5$ and applying the matching template $a_{n-1}X_{n-1}a_n$ from group $T_1$. After recombination, the produced string is $Sa_1X_1a_2X_2 \ldots a_{n-1}X_{n-1}a_n\# = w^{'}$ which is of the form $\{S\}(\Sigma N)^*\{\#\}$ and corresponds to our intended terminal word. At this point, any further recombination at the right most end of this produced terminal string stops because no matching template can be found in the finite set of templates $T$ to guide recombination with this string: 
$$(Sa_1X_1a_2X_2 \ldots a_{n-1}X_{n-1}, X_{n-1} a_n \#) \vdash_{a_{n-1}X_{n-1}a_n} Sa_1X_1a_2X_2 \ldots a_{n-1}X_{n-1}a_n\#.$$

Moreover, for simulating a rule of the form $X_{n-1} \rightarrow \lambda$, the required second string is 
from group $L_6$ and the corresponding template from the group $T_3$. Recombination yields a string $Sa_1X_1a_2X_2 \ldots$\\ $a_{n-1}X_{n-1}\#\#$ of the form $\{S\}(\Sigma N)^*\{\#\#\}$ that is the terminal string and further recombination cannot take place:
$$(Sa_1X_1a_2X_2 \ldots a_{n-1}X_{n-1}, X_{n-1}\#\#) \vdash_{a_{n-1}X_{n-1}\#} Sa_1X_1a_2X_2 \ldots a_{n-1}X_{n-1}\#\#.$$

By construction, it is clear that each string in $\varrho^{*}(L_{0})$ corresponds to a derivation in $G$, and the simulation of a derivation is possible only by using recombinations according to the corresponding template from the finite template language $T$. Accordingly, each 
derivation in $G$ of the form
$$S \Longrightarrow a_1X_1 \Longrightarrow^{*} \ldots \Longrightarrow a_1a_2 \ldots a_{k}X_k \Longrightarrow a_1a_2 \ldots a_{k}a_{k+1}X_{k+1}\Longrightarrow^* \ldots $$
 \hspace*{1.5 in}$a_1a_2 \ldots a_{n-1}X_{n-1} \Longrightarrow	a_1a_2 \ldots a_{n} = w,$\\ 
where $1 \leq k \leq n, X_k \rightarrow a_{k+1}X_{k+1} \in P, X_{n-1} \rightarrow a_{n} \in P,$
corresponds to a computation in $\varrho$ of the form
$$(Sa_1X_1, X_1a_2X_2) \vdash_{a_1X_1a_2} Sa_1X_1a_2X_2 \Longrightarrow^* \ldots Sa_1X_1a_2X_2 \ldots a_{k}X_{k}a_{k+1}X_{k+1}$$ 
$$\Longrightarrow^* \ldots Sa_1X_1a_2X_2 \ldots a_{k}X_{k}a_{k+1}X_{k+1} \ldots a_{n-1}X_{n-1}a_{n}\# = w^{'}, \mbox{ or }$$
 \hspace*{1.0 in}$Sa_1X_1a_2X_2 \ldots a_{k}X_{k}a_{k+1}X_{k+1} \ldots a_{n-1}X_{n-1}\#\# = w^{'}.$

Therefore, we can say from the above description that a terminal string according to the grammar $G$ is achievable only by 
starting the recombination with a string that begins with the start symbol $S$ (that means considering as the first string a string containing the start symbol $S$ at the beginning) and then proceeding by a series of recombination processes according to the appropriate templates from the finite template language $T$ for an arbitrary number of times, which end up with the end marker $\#$ and simulate a derivation according to $G$. Afterwards, intersecting the language $R = \{S\}(\Sigma N)^*\{\# , \#\#\}$ with the set of generated strings, we obtain our intended terminal strings. In this way, we are able to find a string $w^{'} \in \varrho^{*}(L_{0}) \cap R$ and then the application of the weak coding $h(w^{'}) = w \in \Sigma^{*}$ allows us to obtain the exact string generated by a derivation in $G$. Thus, every derivation in $G$ can be simulated. 

Hence, we obtain $L \subseteq h(\varrho^{*}(L_{0}) \cap R)$. The other inclusion follows because the only recombinations that can happen according to $\varrho$ lead either to words that are eliminated by the filter, or to words in $L$ after applying the weak coding.      
\end{proof}
Let us now compare the size of the template language we obtained with that of 
the Daley-McQuillan model \cite{Dal05}. The Daley-McQuillan model \cite{Dal05} requires three production
 rules to construct a template based on their definition of the template language in the following.
$T = \{[X, a, Y][Y, b, Z][Z, c, W]\}$ where $[X, a, Y],[Y, b, Z],[Z, c, W] \in V,$ $V$ is an alphabet and
 $X \rightarrow aY$, 
$Y \rightarrow bZ$, $Z \rightarrow cW \in P$.
If the number of production rules in the grammar is $|P| = n$, then based on this definition the template language
 has a cardinality of $n^3$.  

Our construction requires two production rules to construct a template.
In the worst case we can have $n^2$ templates where $n$ is the number of production rules in the simulated grammar.
 In addition to the size of templates, the size of the TGR alphabet $\Sigma'$ in our construction is small: one
 plus the number of terminals and nonterminals in the simulated grammar. In the Daley-McQuillan model
 as described above, the alphabet $V$ can be much larger, and it also depends on the number of productions 
of the grammar. Although we require fewer templates and alphabet, our model has one limitation, i.e.,
 it requires a filter to discard unintended results, while the Daley-McQuillan model requires
 only the correct recombination to occur according to the constructed matching templates.

\section{CTGR systems with permitting contexts}
\label{sec:Extension_CTGR}
As shown in \cite{Mar06, Mar061, Dal05}, the finiteness of the initial language and the set of
 templates restricts the computational power of a TGR system. In fact, even with a
 regular initial language and a regular set of templates, iterated TGR systems can generate at most regular
 languages \cite{Dal05}.

 Daley and McQuillan \cite{Mar06} have added a new feature called ``deletion context" to enhance the computational
 power of template-guided recombination. Their extension of the TGR system is called the
 \textit{contextual template-guided recombination} system (CTGR system). 
 In \cite{Mar06}, it was shown that arbitrary recursively enumerable languages can be generated by
 iterated CTGR with a regular set of templates and a finite initial language, with the help of taking
  intersection with the Kleene star of the terminal alphabet, and a coding. From a practical viewpoint, dealing
 with a regular set of templates is not realistic in the sense that we cannot manage an infinite ``computer". 

To achieve the finiteness of the employed component sets while preserving the computational power of CTGR,
 we impose an additional control on the templates in order to restrict their usage.
 More precisely, we associate each template with a set of ``permitting contexts'': strings that must  
appear as subwords within the two participating words 
 if this particular template is to be used for their recombination. The idea of permitting contexts has been 
previously used in the context of splicing systems, a formal model of DNA recombination that uses restriction
 enzymes and ligases \cite{Salomaa1998}.

\begin{definition}
\label{CTGRWITHP}
A contextual template-guided recombination system (CTGR system) using permitting contexts is a quadruple $\varrho_p = (T,\Sigma,n_{1},n_{2})$,
where $\Sigma$ is a finite alphabet, $n_{1}\in \mathbb{N}$ is the minimum MDS length and $n_{2}\in \mathbb{N}$ is the minimum pointer length, $T$ is a set of triples (templates using permitting contexts) of the form $t_p = (t; C_1, C_2)$ with $t = e_1\#\alpha\beta\gamma\#d_1$ being a template over $\Sigma$ and $C_1, C_2$ being finite subsets of $\Sigma^{*}$. To such a triple $t_p$ we associate the word\\
\hspace*{.5 in} $\tau(t_p)= e_1\#\alpha\beta\gamma\#d_1\$a_1\&\ldots\&a_k\$b_1\& \ldots \&b_m$ ,\\
 where $C_1 = \{a_1,\ldots,a_k\}$, $C_2 = \{b_1,\ldots,b_m\}$, $k,m \geq 0$ and $\$, \&,\#$ are new special symbols not included in $\Sigma$. We define $\tau(T) = \{\tau(t_p)\mid t_p \in T\}$.  

For a CTGR system using permitting contexts $\varrho_p =(T,\Sigma,n_{1},n_{2})$ and a language $L\subseteq \Sigma^{*}$, we define $\varrho_p(L)=\{w\in\Sigma^{*}\mid(x,y)\vdash^{c}_{t_p}w$ for some $x,y \in L$ and $t_p \in T\}$, where $(x,y)\vdash^{c}_{t_p}w$ if and only if $x = u\alpha\beta d_1 d$, $y = e e_1\beta\gamma v$, $t_p = (e_1\#\alpha\beta\gamma\#d_1; \{a_1, \ldots,a_k\}, \{b_1, \ldots ,b_m\})$, $w = u\alpha\beta\gamma v, u,v,d,e \in \Sigma^{*}$, $\alpha,\gamma \in \Sigma^{\geq n_1}$, $\beta \in \Sigma^{\geq n_2}$. Every element that belongs to $C_1$ appears as a substring in $x$ and every element that belongs to $C_2$ appears as a substring in $y$, i.e., $a_i \in sub (x)$ for $1 \leq i \leq k$, $b_j \in sub (y)$ for $1 \leq j \leq m$; moreover, if $C_1 = \{\lambda\}$ or $C_2 = \{\lambda\}$, then we assume that no constrain is imposed on $x$ and $y$ respectively.  \end{definition}

For a \textit{CTGR system using permitting contexts} $\varrho_p =(T,\Sigma,n_{1},n_{2})$ and a language $L \subseteq \Sigma^{*}$, a template language $T$, we can define an iterated version of $\varrho^{*}_p(L)$ similarly as for TGR systems.

The following proposition shows that iterated \textit{CTGR system using permitting contexts} can generate arbitrary recursively enumerable languages using a finite initial language and a finite set of templates with the help of intersection with a filter language and, at last, applying a weak coding homomorphism.

\begin{proposition}
\label{CTGRPRE}
Every recursively enumerable language $L \subseteq \Sigma^{*}$ can be written in the form $L = h(L^{'} \cap L_1)$, where $h$ is a weak coding homomorphism, $L_1$ is a regular language and $L^{'} = \varrho^{*}_p(L_0)$ with $L_0$ a finite language. 
\end{proposition} 

\begin{proof}
Consider a Chomsky type-0 grammar $G = (N,\Sigma,S,P)$ in Kuroda normal form, where $L(G)= L$ and the
 production rules in $P$ are of the forms
 $A \rightarrow EC$, $AE \rightarrow CD$, $A \rightarrow a \mid \lambda$ for $A, C, D, E \in N$, $a \in \Sigma$.
 Let us denote $U = N \cup \Sigma \cup \{B, B_1, B_2\}$, where $B, B_1, B_2$ are new symbols.

We then construct a CTGR system \textit{using permitting contexts} $\varrho_p = (T,V,1,1)$\\
 where $V = N \cup \Sigma \cup \{B,B_1,B_2, X, X^{'}, Y, Z, Z^{'}\} \cup \{Y_{b}\mid b \in U\}$
 \begin{tabbing}
 and \=$T$ \=co\=ntains th\=e following templates using permitting contexts:\\
 \>\> Simulate : 1. $Z\#cavY\#uY; \{X\}, \{\lambda\}$,   for $a,c \in U, Z, Y \in V, u \rightarrow v \in P$,\\
 \>\>\> Rotate : 2. $Z\#caY_b\#bY; \{X\}, \{\lambda\}$,  for $a,b,c \in U, Z, Y \in V$,\\	
 \>\>\>\>		 	 3. $X\#X^{'}bde\#Z; \{\lambda\}, \{Y_b\}$,	  for $b,d,e \in U, Z, X \in V$,\\
 \>\>\>\>		 	 4. $Z\#caY\#Y_b; \{X^{'}\}, \{\lambda\}$, for $Z, Y, Y_b \in V$,\\
 \>\>\>\>		 	 5. $X^{'}\#Xac\#Z; \{\lambda\}, \{Y\}$, for $X^{'}, X, Z \in V,$\\
 \hspace*{.3 in} Terminate : 6. $XBB_1B_2\#abc\#Z^{'}; \{\lambda\}, \{Y\}$, for $X,Z^{'} \in V, B,B_1,B_2 \in U$.
\end{tabbing}
\noindent
We define the following languages which are included in the initial finite language $L_0 \subseteq V^{*}$:

\noindent
\hspace*{.5 in} $L_1 = \{XBB_1B_2SY\}$, $L_2 = \{ZavY \mid a \in U, u \rightarrow v \in P\}$,\\
\hspace*{.5 in} $L_3 = \{ ZaY_b \mid a \in U \}$, $L_4 = \{X^{'}baZ \mid b,a \in U \}$,\\
\hspace*{.5 in} $L_5 = \{ZaY \mid a\in U\}$, $L_6 = \{XaZ \mid a \}$, $L_7 = \{abZ^{'} \mid a,b \in U\}$.\\
\noindent
We denote $L_0 = L_{1} \cup L_2 \cup L_3 \cup L_4 \cup L_5 \cup L_6 \cup L_7$, and
 $L_0$ acts as the initial language.

For the construction of this system, the idea we use is the well-known proof technique, 
``rotate-and-simulate procedure", which was effectively used in other contexts \cite{Salomaa1998} in order 
to allow the simulation of a rule that applies to a symbol in the middle of the word by first moving that symbol
 to the right hand end of the word, simulating the rule, and returning the result to its original place.

Throughout this construction we assume $x$ and $y$ to be, respectively, the first word and the
 second word of the recombination as defined in Definition ~\ref{CTGRWITHP}.

The starting of the simulation based on the derivation steps in $G$ requires to consider the word $XBB_1B_2SY$ as the first word. Indeed, any other choice of start word leads to derivations of words of illegal form (not in $XBB_1B_2\Sigma^{*}Y$). Throughout the derivation steps this word is bordered by $X$ or its variant $X^{'}$ at the left end, as well as by $Y$ or its variant $Y_{b}$, $b \in U$ at the right end, $X$, $X^{'}$ and $Y$, $Y_b$ make the left respectively right extremity of the word. Likewise, the symbol $B$ always signals the beginning of the word, i.e., the sentential forms of $G$, which facilitates the permutation of the word and $B_1, B_2$ are included to provide the contexts for recombination.

Note that all the templates with permitting contexts in $T$ include symbols $Z$ or $Z^{'}$ that have thus to be present in one of the two words taking part in the recombination. Furthermore, the words containing symbols $Z$ and $Z^{'}$ are from the initial language $L_0$ but will not appear in the resulting word of recombination. This guarantees that each recombination has to happen between the current word which is produced in the previous recombination and at least one word from $L_0$. The simulation of a derivation in $G$ initiates with the application of template 1 to $XBB_1B_2SY \in L_1$ and $ZavY \in L_2$, where initially $w = BB_1B_2S$, $S \rightarrow v \in P$ and $a \in U$. The word obtained through the recombination is $XBB_1B_2vY$:  
$$(XBB_1B_2SY, ZB_2vY)\vdash_{t_p} XBB_1B_2vY$$
for $t_{p} = (Z\#cavY\#uY; \{X\}, \{\lambda\})= (Z\#B_1B_2vY\#SY; \{X\}, \{\lambda\}),$ where $S \rightarrow v \in P$, $c, a\in U$ and $w = BB_1B_2S$.

Generally, considering a word $Xx_1BB_1B_2x_2uY$ and $u \rightarrow v \in P$, the resulting word will be\\ $Xx_1BB_1B_2x_2vY$ applying the associated templates from group 1. Here, $w = w_1cau = x_1BB_1B_2x_2u$. This simulates a derivation step $x_2ux_1 \Longrightarrow x_2vx_1$ in $G$. The derivation is as follows: 
$$(Xw_1cauY, ZavY)\vdash_{t_p} Xw_1cavY$$
for $t_{p} = (Z\#cavY\#uY; \{X\}, \{\lambda\}),$ where $u \rightarrow v \in P$, $c, a \in U$.

In this simulation step, no other templates from groups 2 - 6 can be applied except the templates from group 1 because of imposed restriction as deletion contexts and permitting contexts on the usage of the templates. Afterwards, we come to the rotation process that is necessary so as to move symbols from the right hand end of the current word to the left hand end. This rotation process can be explained by the following steps:

\textit{Step 1}: We can start the rotation process using the corresponding template from group 2 with a word $XwbY$, where $b \in (N \cup \Sigma)^{*}, w \in (N \cup \Sigma)^{*} \{BB_1B_2\}(N \cup \Sigma)^{*}$ (respectively $XwbY$, where $b \in \{B, B_1,B_2\}, w \in (N \cup \Sigma)^{*}$). In this step, $x = XwbY = Xw_1cabY$, $y = ZaY_{b} \in L_3$: 
$$(Xw_1cabY, ZaY_{b})\vdash_{t_p} Xw_1caY_{b}$$ 
for $t_{p} = (Z\#caY_b\#bY; \{X\}, \{\lambda\}),$ where $wb \in (N \cup \Sigma)^{*}\{BB_1B_2\}(N \cup \Sigma)^{*}, b \in N \cup \Sigma\cup \{B,B_1,B_2\}$.

\textit{Step 2}: After applying the template from group 2 in Step 1, we obtained the word $Xw_1caY_{b}$, which we rewrite as $Xdew_2Y_b$ where $w = w_1ca = dew_2$. Then we continue the rotation process using the matching template from group 3 with $x = X^{'}bdZ \in L_4$, $y = Xdew_2Y_b$:
$$(X^{'}bdZ, Xdew_2Y_b)\vdash_{t_p} X^{'}bdew_2Y_b$$
for $t_{p} = (X\#X^{'}bde\#Z; \{\lambda\}, \{Y_b\}),$ where $b \in N \cup \Sigma \cup \{B,B_1,B_2\}$.

\textit{Step 3}: The resulting word from the previous step is of the form $X^{'}bdew_2Y_b$, which can be written of the form $X^{'}wY_b = X^{'}w_3caY_b$. In this step we will apply the matching template from group 4 where $x = X^{'}w_3caY_b$, $y = ZaY \in L_5$:
$$(X^{'}w_3caY_b, ZaY)\vdash_{t_p} X^{'}w_3caY$$ 
for $t_{p} = (Z\#caY\#Y_b; \{X^{'}\}, \{\lambda\})$ where $c, a \in N \cup \Sigma \cup \{B,B_1,B_2\}$.

\textit{Step 4}: The recombined word from Step 3 is $X^{'}w_3caY$, in general, the outcome of step 3 is a word of the form $X^{'}acw_4Y$. Lastly, we complete the rotation process by using a template from group 5 where $x = XaZ \in L_6$, $y = X^{'}acw_4Y$: 
$$(XaZ, X^{'}acw_4Y)\vdash_{t_p} Xacw_4Y$$
for $t_{p} = (X^{'}\#Xac\#Z; \{\lambda\}, \{Y\})$ where $a, c \in N \cup \Sigma \cup \{B,B_1,B_2\}$.

The above mentioned rotation-steps produced the word $Xbacw_4Y = XbwY$ which implies that starting from the word $XwbY$ and applying steps 1 - 4, we achieve the word $XbwY$ having the same end markers. In this way, we are able to move the symbol $b$ from the right-hand end to the left-hand end of the word that accomplishes the rotation of the underlying sentential form. These rotation-steps can be repeated an arbitrary number of times and thus provide every circular permutation of the word flanked by $X$ and $Y$. 

Using a template from group 1 to each word of the form $XwY$ when $w$ ended by the left hand part of a rule in $P$, it is possible to simulate the application of all rules of $P$ at a desired position corresponding to the sentential form of $G$, by means of the four rotation-steps. 

It is observed that from the initial word $XBB_1B_2SY$, each produced word in every step does not include the symbols $Z$, $Z^{'}$, that is, the word is of the form $\alpha_1x_1BB_1B_2x_2\alpha_2$ in which the pair $(\alpha_1, \alpha_2)$ is one of the four pairs $(X, Y)$, $(X, Y_b)$, $(X^{'}, Y_b)$, $(X^{'}, Y), b \in U$. In fact, these symbols being present in the templates of $T$ serve as permitting contexts that restrict the regulation of the recombination process of this system $\varrho_p$.

Now we come to the termination process. Applying the terminating template from group 6, we can remove $XBB_1B_2$ only when $Y$ is present and the symbols $B,B_1,B_2$ together as a word $BB_1B_2$ is adjacent to $X$. Here, $x = abZ^{'} \in L_7$, $y = Xacw_4Y = XBB_1B_2bcw_5Y:$
$$(abZ^{'}, XBB_1B_2bcw_5Y)\vdash_{t_p} abcw_5Y$$
for $t_{p} = (XBB_1B_2\#abc\#Z^{'}; \{\lambda\}, \{Y\}),$ where $w \in (N \cup \Sigma)^{*}\{BB_1B_2\}(N \cup \Sigma)^{*}$, $b, a, c \in N \cup \Sigma \cup \{B,$\\ $B_1, B_2\}$.  

Now our achieved word is of the form $abcw_5Y = wY = w^{'} \in \varrho^{*}_p(L_0)$, i.e., $L^{'} = \varrho^{*}_p(L_0)$. The intersection operation with the language $L_1 = \Sigma^{*}Y$ will filter out the words that are not in proper form. Furthermore, we define a weak coding homomorphism which eliminate the right end marker $Y$ leaving other letters unchanged. Let us now define a weak coding
 homomorphism by
$h(a) = a$ , for any $a \in \Sigma,$
$h(Y) = \lambda$.

Thus, we obtain a word in $\Sigma^{*}$ by applying the weak coding homomorphism where $w \in h(L^{'} \cap L_1)$. Finally, from the above construction we can produce each word in $L(G)$ and we say that $L(G) \subseteq h(\varrho^{*}_p(L_0) \cap \Sigma^{*}Y)$. Conversely, the opposite inclusion is held by this system. Therefore, $h(\varrho^{*}_p(L_0) \cap \Sigma^{*}Y) \subseteq L(G)$.
\end{proof}

\section{Conclusions}
\label{sec:Conclusion}
This paper improves on the descriptional complexity (size of the template language) from $n^3$ to $n^2$ in 
the case of template-guided recombination (TGR) systems, and from regular to finite in the case of contextual 
template-guided recombination (CTGR) systems. These reductions are obtained at the expense of using a filtering
 language in the case of TGR, and of an additional control (permitting contexts) in the case of CTGR.

\bibliography{references_database} 
\bibliographystyle{eptcs}
\end{document}